\documentclass[11pt]{article}

\usepackage[alg]{KLMM/klmm}

\newcommand{\ZZ}{\mathbb Z}
\newcommand{\Z}{\mathbb Z}
\newcommand{\Fq}{\mathbb F_q}
\newcommand{\SoftO}{\widetilde O}
\newcommand{\softO}{\widetilde{O}}
\newcommand{\polylog}{\operatorname{polylog}}
\newcommand{\supp}{\operatorname{supp}}
\newcommand{\charac}{\operatorname{char}}
\newcommand{\norm}[1]{\left\lVert #1\right\rVert}
\newcommand{\Interpolate}{\textsc{Interpolate\_mbb}}

\def\R{{\mathbb{R}}}

\def\Z{{\mathbb{Z}}}

\def\R{{\mathcal{R}}}
\def\F{{\mathbb{F}}}

\def\Pr{{\mathbb{P}}}

\def\polylog{\hbox{\rm{polylog}}}

\usepackage{xcolor}

\title{Quasi-linear Time Multiplication of Sparse Polynomials with Integer Coefficients}
\author{
    Qiao-Long Huang\institution{Shandong University}
    \and Yichuan Cao\institution{State Key Laboratory of Mathematical Sciences, Academy of Mathematics and Systems Science, Chinese Academy of Sciences; University of Chinese Academy of Sciences}
    \and Ruichen Qiu\instref{2}
    \and Xiao-Shan Gao\instref{2}\Cauthor
}

\begin{document}
\maketitle

\begin{abstract}
\noindent
Sparse polynomial multiplication is a fundamental problem in computer algebra and the theory of computation, and the development of a quasi-linear time output-sensitive multiplication algorithm has been posed as an open challenge.
In this paper, a counterexample is provided to a previously claimed solution to this open problem for integer coefficients.
By employing the existing quasi-linear modular-black-box interpolation algorithm, we are able to provide an algorithm  with quasi-linear bit complexity for the integer coefficients setting.
Furthermore, in the case of coefficients over a finite field, we obtain an algorithm whose bit complexity is linear in the number of terms, the logarithm of the degree, and the logarithm of the size of the finite field.
\end{abstract}

\section{Introduction}
\label{sec-1}

Polynomial multiplication is a fundamental algebraic operation that appears in many scientific and engineering computations. Consequently, it is also a basic issue in computer algebra and theory of computation \cite{ModernCA,Jin2025STOC,NickFischer2025SODA}.
The classical multiplication algorithms for dense polynomials ~\cite{KaratsubaOfman1962,Toom1963,Cook1966,SchonhageStrassen1971,CantorKaltofen1991} run in quasi-linear time with respect to the polynomial degree and are therefore optimal, up to polylogarithmic factors.
%

In practice—and in computer algebra systems like Maple and Mathematica—polynomials are usually represented in a sparse format~\cite{SparsityChallenges}, meaning they are a list of their nonzero coefficients together with the associated exponents.
For a sparse univariate polynomial $f$ of degree $D$, containing $T$ nonzero terms and whose coefficients have a maximum absolute value of $H$, its bit size is tightly bounded by
$$\text{size}(f) = O(T(\log D+\log H))=O(T\log(DH)).$$
The challenge in sparse polynomial multiplication is that the product is output-sensitive; that is, the size of the result cannot be deduced solely from the sizes of the inputs can can be exponential in the input size, as shown by the following example.

\begin{example}[\cite{Grenet-H2026}]
\label{ex-01}
Let $f = \sum_{i=0}^{d-1} x^i$, $g_1 = \sum_{i=0}^{d-1} (x^{id+1} - x^{id})$ and $g_2 = \sum_{i=0}^{d-1} (x^{id+1} + x^{id})$, in $\mathbb{Z}[x]$.
Note that $g_1$ and $g_2$ have the same support. However, $f  g_1 = x^{d^2} - 1$ has  two terms, and $f  g_2 = x^{d^2} + \sum_{i=1}^{d^2-1} 2x^i + 1$ has $d^2+1$ terms.
\end{example}
Therefore, the dense multiplication algorithms fail to accurately reflect the true complexity of the sparse setting, and meaningful sparse multiplication algorithms must be output-sensitive. 
Open Problem 1 in \cite{Roche2018} calls for the development of a quasi-linear, output-sensitive algorithm for sparse polynomial multiplication over an arbitrary ring $\R$:

\begin{quote}
\textbf{Open Problem 1.} Develop an algorithm to multiply two sparse polynomials $f, g \in \R[x]\setminus\R$  using $ \widetilde{O}(T \log D)$  ring and bit operations, where $T$ is an upper bound for the number of terms in $f$,  $g$,  $fg$,  and $D$  is an upper bound on their degrees.
\end{quote}
If the coefficients  are integers, Open Problem 1 can be rephrased more explicitly as follows.

\begin{quote}
\textbf{Open Problem 2.} Develop an algorithm to multiply two sparse polynomials $f, g \in \Z[x]\setminus\Z$  using $\widetilde{O}(T \log (DH))$ bit operations,  
where $D = \max(\deg(f),\deg(g))$, $H = \max(\|f\|_\infty, \|g\|_\infty, \|fg\|_\infty)$, and $T = \max(\|f\|_0, \|g\|_0, \|fg\|_0))$.
\end{quote}
Since $B = T \log(DH)$ gives a tight upper bound on the bit-size of $f, g,$ and $fg$, an algorithm whose complexity is quasi-linear in $B$ is optimal, up to polylogarithmic factors in $B$.

There has been significant progress in sparse polynomial multiplication~\cite{Thesis-Arnold,9076060,Giorgi2020-ISSAC,GIORGI202398-JSC}.
In \cite{Giorgi2020-ISSAC}, the authors assert that Open Problem 2 has been positively resolved.
However, this paper presents a counterexample to a key result, Lemma 4.7(ii), of that work, thereby invalidating the claim.

In this paper, by employing the quasi-linear modular-black-box (MBB) interpolation algorithm from \cite{Giorgi2022-ISSAC}, we are able to confirm that Open Problem 2 holds true.

We further propose a sparse multiplication algorithm over the finite field $\F_q$, which relies on the sparse interpolation method from \cite{Huang2019-ISSAC}, achieving a bit complexity of $\widetilde{O}(T \log D \log q)$ when the characteristic of the finite field is large. This is the strongest result to the best of our knowledge.

Over arbitrary rings, Arnold presents a randomized output-sensitive algorithm that performs $\widetilde{O}(T\log^3 D)$ ring operations~\cite{Thesis-Arnold}, which, to the best of our knowledge, remains the strongest known result.

\subsection{Main Results}

The following multivariate version of Open Problem 2 has been proven.
\begin{theorem}
\label{th-m1}
There exists an algorithm that accepts two sparse polynomials $f, g \in \Z[x_1,\ldots,x_n]\setminus\Z$ as input and gives the product $fg$ using 
$\widetilde{O}(T(n\log(D)+ \log(H))$ bit operations,  
where 
\[
D = \max_{i=1}^n\{\deg(f,x_i),\deg(g,x_i)\}, 
H = \max(\|f\|_\infty, \|g\|_\infty), 
T = \max(\|f\|_0, \|g\|_0, \|fg\|_0)).
\]
\end{theorem}

The proof consists of three steps.
First, we present a quasi-linear algorithm for univariate polynomials under the assumption that the output sparsity is known. 
Different from \cite{Giorgi2020-ISSAC}, Theorem \ref{th-m1} uses
MBB sparse interpolation theorem given in \cite{Giorgi2022-ISSAC}.  Once valid degree, height, and sparsity bounds for the unknown product are supplied, multiplication can be implemented by answering the interpolation oracle through evaluations of the
two input polynomials. 
This gives the supplied-output-sparsity form of Theorem \ref{th-m1}. 
When the output sparsity is not supplied, interpolation alone does not provide a stopping criterion.
Classic product verification methods and Krnonecker substitution in
\cite{Giorgi2020-ISSAC} are then used to certify the guessed product and to provide a final proof of Theorem \ref{th-m1}.

Over the finite fields, the following result is proven.
\begin{theorem}
\label{th-m2}
For a finite field $\Fq$ with characteristic no less than $O(D^n)$, there exists an algorithm that accepts two sparse polynomials $f, g \in \F_q[x_1,\ldots,x_n]\setminus\F_q$ as input and gives the product $fg$ using 
$\widetilde{O}(Tn\log(D)\log(q))$ bit operations,  
where 
$$
D = \max_{i=1}^n\{\deg(f,x_i),\deg(g,x_i)\},
T = \max(\|f\|_0, \|g\|_0, \|fg\|_0)).$$
\end{theorem}

Theorem \ref{th-m2} follows a different route from that of Theorem \ref{th-m1}.  
It uses sparse interpolation over large-characteristic fields (\({\rm char}(\Fq)>O(D^n)\)) in the style of Huang \cite{Huang2019-ISSAC}: derivative images produce candidate exponents, and ordinary cyclic images filter the candidates.  
Dense cyclic products are charged using the Cantor-Kaltofen multiplication theorem \cite{CantorKaltofen1991}. 
For multivariate cases, directly applying the classical Kronecker substitution converts a $D$-degree multivariate polynomial into a univariate polynomial of degree $D^n$. The subsequent computation relies on differentiation, which requires the characteristic of the finite field to be greater than the degree of the univariate polynomial, that is, $\operatorname{char}(\mathbb{F_q}) > D^n$. For the univariate case, the requirement reduces to $\operatorname{char}(\mathbb{F}) > D$, which is much more reasonable.


The main results presented in this paper are obtained through an interactive collaboration between the authors and an artificial intelligence agent system, \textit{MechMath Agent Team (MMAT)}~\cite{MMAT}. 
The authors assume full responsibility for the paper’s content.
MMAT is a large language model driven agent designed to prove mathematical theorems expressed in both natural language and formal language in Lean. 

%
%
%
%

\subsection{Lean Formalization}
\label{subsec:lean-formalization}
Using the agent system MechMath Agent Team (MMAT)~\cite{MMAT}, we also developed a Lean formalization of the main algebraic and complexity claims with Lean~4.29.0~\footnote{\url{https://github.com/EonMath/sparse-product}}.
The formalization covers the supplied-output-sparsity integer theorem, Theorem~\ref{th-ql1}, and the large-characteristic finite-field sparse-product theorem, Theorem~\ref{thm:fq-product}, by proving the internal algebraic and cost-accounting lemmas and treating the large external algorithmic results as named axioms.  
It also records the univariate integer reduction theorem
\[
  \texttt{SparseProduct.SparseProductCyclicImage.integerSparseProductReduction},
\]
which is the Lean counterpart of the finite-field route for integer coefficients. 
The multivariate Kronecker-substitution wrapper and the product-verification wrapper for unknown output sparsity are not yet formalized, with only the univariate supplied-output and finite-field core theorems exposed in \texttt{SparseProduct.lean}.

For Theorem~\ref{th-ql1}, the Lean theorem \texttt{SuppliedOutputSparsity.suppliedOutputSparsityTheorem} is the formal statement of the displayed \(\widetilde O(T\log(DH))\) bit-complexity bound.
The supporting Lean file \texttt{SparseProduct/SuppliedOutputSparsity.lean} formalizes the monomial case, product height bound, degree decomposition, modular black-box correctness, shifted-root identity, shifted evaluation identity, and geometric probe-cost summation.
The two external interpolation facts from Giorgi--Grenet--Perret du Cray--Roche~\citep{Giorgi2022-ISSAC} are imported as the named axioms \texttt{ggprTheorem34} and \texttt{ggprFact32}.

For Theorem~\ref{thm:fq-product}, the Lean theorem \texttt{SparseProduct.SparseProductCyclicImage.finiteFieldSparseProduct} is the top-level finite-field statement.
The file \texttt{SparseProduct/SparseProductCyclicImage.lean} proves the cyclic-image identities, derivative candidate construction, bad-prime counting,
prime-isolation scaffold, and candidate-filtering scaffold.
The final Monte Carlo finite-field and integer sparse-product statements rely on the named randomized-algorithm axioms in \texttt{SparseProduct/RandomizedAlgorithms.lean}; three additional local axioms package the remaining prime-density and uniform finite-set sampling probability steps~\citep{agrawal2004primes,CantorKaltofen1991,HardyWright2008}.
Thus, the retained development has no anonymous \texttt{sorry} holes; its unproved assumptions are explicitly named axioms.  
Further details are given in Appendix~\ref{app:lean-formalization}.

\subsection{Related Work}

Polynomial multiplication has followed three rather different lines: dense multiplication, positive coefficient multiplication, and sparse multiplication.  

The classical multiplication algorithm for dense univariate polynomials is quadratic in degree. 
The first subquadratic divide-and-conquer methods go back to
Karatsuba and Ofman \cite{KaratsubaOfman1962}; Toom--Cook multiplication
\cite{Toom1963,Cook1966} generalizes the same idea by evaluating the input polynomials at several points, multiplying the values, and interpolating.
Finally, the Schonhage-Strassen method
\cite{SchonhageStrassen1971} and the Cantor-Kaltofen method \cite{CantorKaltofen1991} provide quasi-linear algorithms
using FFT when suitable roots of unity are available. 

Positive coefficients lead to a distinct output-sensitive regime.
For $f,g\in\mathbb Z_{+}[x]$ with $T=\#(fg)$, $H = \max(\|f\|_\infty, \|g\|_\infty)$, and degree bound $D$,
Cole and Hariharan~\cite{Cole2002STOC}
obtain an algorithm with complexity $O(T \log D)$;
Jin and Xu prove that $fg$ can be computed by a Las Vegas randomized algorithm using $O(T\log H)$ word operations~\cite{Jin2025STOC}. 
%
Fischer obtains the corresponding linear bound in a real-RAM model~\cite{NickFischer2025SODA}. 
As illustrated in Example \ref{ex-01}, these findings do not directly resolve the complexity question in the integer case with signed coefficients, since there are no term cancelations in the positive coefficient case.

For the sparse setting, Arnold's interpolation-and-testing framework provides a randomized output-sensitive sparse multiplication algorithm using $\tilde{O}(T\log^3 D)$ ring operations~\cite{Thesis-Arnold}.
Over $\mathbb{Z}$, Nakos provides an algorithm with quasi-linear complexity in word operations~\cite{9076060}. However, its bit complexity is $\tilde{O}(T\log^2D\log(HD))$~\cite{Giorgi2020-ISSAC}. 
Giorgi, Grenet, and Perret du Cray \cite{Giorgi2020-ISSAC} claimed a quasi-linear algorithm over the integers, but Section~\ref{sec-ce} presents a counterexample.
%
A modular product verification algorithm is provided~\cite{GIORGI202398-JSC}.

The rest of the paper is organized as follows.
In Section \ref{sec-ce}, a counterexample to Lemma 4.7(ii)   of \cite{Giorgi2020-ISSAC} is presented.
In Section \ref{sec-Mult-I}, the multiplication algorithm over integers is given.
In Section \ref{sec-Mult-I}, the multiplication algorithm over finite fields is given.
In Section \ref{sec-Conc}, conclusions are presented.

\section{The Failed Probability Claim}
\label{sec-ce}
In this section,  a counterexample is presented to Lemma 4.7(ii), thereby invalidating the complexity claim of the multiplication algorithm in \cite{Giorgi2020-ISSAC}.

For \(F,G\in\ZZ[X]\), the previously claimed result used, in Algorithm 2 Step 3~\cite{Giorgi2020-ISSAC}, a prime \(p\) sampled
from an interval and asserted that
\[
T_p=\max\{\#(F_pG_p),\#(F_pG'_p+F'_pG_p)\}\le \#(FG)
\]
with probability at least \(1-\mu_1/2\).  Here \(F_p,G_p,F'_p,G'_p\) denotes the
ordinary reduced polynomials used by the interpolation calls before the later
reduction modulo \(X^p-1\).  The following example refutes this probability
claim.


\begin{example}
\label{ex-counter}
        Let
\[
\mu_1=\frac12,\qquad N=10^{20},\qquad
F=(1-X^N)^4,\qquad G=(1+X^N)^4.
\]
Then \(F,G\in\ZZ[X]\), and
\[
FG=(1-X^{2N})^4.
\]
Therefore
\[
\#F=\#G=\#(FG)=5,\qquad D=\deg(FG)=8N.
\]
\end{example}

For this input, Step 3 of Algorithm 2 yields
\[
\lambda=\frac{20}{3\mu_1}(\#F\#G)^2\ln D
       =\frac{25000}{3}\ln(8\cdot10^{20})
       =401092.861679672913\ldots .
\]
Hence the integer primes in \([\lambda,2\lambda]\) are exactly the primes in
\[
[401093,802185].
\]
There are \(30182\) such primes.  Exact enumeration over this prime interval
gives the support-count table:
\begin{table}[h]
    \small
    \centering
    \begin{tabular}{c c c c}
        \toprule
        \#$(F_pG_p)$ & \#$(F_pG'_p+F'_pG_p)$ & $T_p$ & number of primes\\
        \midrule
        5 & 4 & 5 & 7495\\
        7 & 6 & 7 & 17698\\
        8 & 7 & 8 & 4989\\
        \bottomrule
    \end{tabular}
    \label{tab:supp-count}
\end{table}

Thus
\[
\Pr[T_p\le \#(FG)]
=\frac{7495}{30182}
=0.248326817308\ldots
<\frac34
=1-\frac{\mu_1}{2}.
\]
This contradicts the probability claim in Lemma 4.7 (ii).

The mechanism is that cancellations occurring in
\[
FG=(1-X^{2N})^4
\]
can split across different ordinary exponents after \(F\) and \(G\) are
reduced separately.  For example, for \(p=401101\),
\[
N\bmod p=277939
\]
and
\[
\begin{aligned}
F_pG_p={}&1-16X^{63230}+12X^{154777}+6X^{309554}
+12X^{464331}-16X^{555878}+X^{619108}.
\end{aligned}
\]
So \(\#(F_pG_p)=7>5=\#(FG)\).

The conclusion is limited to the written proof step.  Lemma 4.7(ii)'s
load-bearing probability statement is false under the uniform interval-prime
model used in that collision-counting argument.  Consequently, the published
proof of Theorem 1.1 does not establish the claimed output-sensitive complexity
bound over \(\ZZ\).  This does not rule out every possible repaired algorithm or alternative analysis.~\footnote{
We write a Maple program to present the details of the example, which is available in \url{https://github.com/EonMath/sparse-product}.
}

\section{Multiplication with Integer Coefficients}
\label{sec-Mult-I}
In this section, we prove Theorem \ref{th-m1}.
The supplied output sparsity case is presented as Theorem \ref{th-ql1} in Section \ref{sec-31}.
The general case is presented in Section \ref{sec-32} based on Theorem \ref{th-ql1}.

\subsection{Supplied Output-sparsity Multiplication}
\label{sec-31}
The integer result answers the supplied-output-sparsity form of Open
Problem 2.  The problem asks for an algorithm to multiply two sparse
polynomials \(f,g\in \Z[x]\) using
\[
  \softO\bigl(T\log(DH)\bigr)
\]
bit operations, where
\[
  D=\deg(fg),\qquad
  H=\max(\norm{f}_{\infty},\norm{g}_{\infty},\norm{fg}_{\infty}),
  \qquad
  T=\max(\norm{f}_{0},\norm{g}_{0},\norm{fg}_{0}).
\]
Here \(\norm{\cdot}_{0}\) denotes sparsity and \(\norm{\cdot}_{\infty}\)
denotes the largest absolute coefficient.  Changing the logarithm base changes
only absolute constants, which are absorbed by the soft-O notation.

\begin{theorem}[Supplied output sparsity]
\label{th-ql1}
Assume that the sparsity parameter \(T\) in Open Problem 2 is supplied to the
algorithm.  There is a Monte Carlo algorithm which returns the sparse
representation of \(fg\) with probability at least \(2/3\) and uses
\[
  \softO\bigl(T\log(DH)\bigr)
\]
bit operations.
\end{theorem}

\paragraph{Source theorem used.}
The proof uses Theorem 3.4 of Giorgi, Grenet, Perret du Cray, and
Roche~\cite{Giorgi2022-ISSAC}, in its modular-black-box interpolation form.  That
theorem supplies a Monte Carlo interpolation algorithm for an unknown integer
polynomial once valid degree, height, and sparsity bounds are supplied.  The
implementation of the black-box probes also uses Fact 3.2 of the same
paper~\cite{Giorgi2022-ISSAC}, which gives the batched root-of-unity evaluation used
below.  The optional discussion in which the output sparsity is not supplied
requires an additional sparse product verification wrapper, as described by \citet{GIORGI202398-JSC}.


\begin{algorithm}[H]
\caption{Multiplication of sparse polynomial with integer coefficients}
\begin{algorithmic}[1]
  \Require Sparse polynomials \(f,g\in\mathbb{Z}[x]\), a supplied sparsity bound
  \(T\ge \max(\Vert f\Vert_{0},\Vert g\Vert_{0},\Vert fg\Vert_{0})\), and
  \(D=\deg(fg)\ge 2\).
  
  \Ensure The sparse representation of \(fg\) with probability at least
  \(2/3\).
  
  \State If \(T=1\), multiply the two input monomials directly and return the result.
  
  \State Write \(f=\sum_{a\in A} f_ax^a\) and \(g=\sum_{b\in B}g_bx^b\). Set \(h=fg\) as the unknown target polynomial and compute
  \[
    H_0=\max\{2,\Vert f\Vert_{0}\Vert g\Vert_{0} \Vert f\Vert_{\infty}\Vert g\Vert_{\infty}\}.
  \]
  
  \State Give \(\text{Interpolate}\) the degree bound \(D\), height bound \(H_0\), and sparsity bound \(T\).
  
  \State When \(\text{Interpolate}\) asks for a modular black-box value at \((\theta,m)\), return
  \[
    \mathcal B_h(\theta,m)= \left(\sum_{a\in A} f_a\theta^a\right) \left(\sum_{b\in B} g_b\theta^b\right)\bmod m .
  \]
  
  \State Answer the root-of-unity probe lists used by \(\text{Interpolate}\) by batched sparse evaluations of \(f\) and \(g\), followed by pointwise multiplication in the ambient ring.
  
  \State Return the polynomial produced by \(\text{Interpolate}\).
\end{algorithmic}
\end{algorithm}

\begin{proof}
Put \(h=fg\).  Since \(D=\deg h\ge 2\), the product is nonzero and nonconstant,
and, because \(\Z[x]\) is an integral domain,
\[
  D=\deg f+\deg g .
\]
The sparse inputs therefore determine the degree bound \(D\).

First consider the possible monomial case \(T=1\).  Then both inputs have one
term, say \(f=ax^r\) and \(g=bx^s\).  The algorithm returns
\[
  ab\,x^{r+s}.
\]
The coefficient multiplication costs \(\softO(\log(H))\) bit operations,
and writing the exponent costs \(O(\log(D))\), where \(H = \max\bigl(\|f\|_{\infty}, \|g\|_{\infty}, \|fg\|_{\infty}\bigr)\).  
This is within \(\softO(T\log(DH))\).

It remains to treat \(T\ge 2\).  Write
\[
  f=\sum_{a\in A} f_a x^a,\qquad
  g=\sum_{b\in B} g_b x^b
\]
with nonzero displayed coefficients.  Define the input-computable height bound
\[
  \widehat H
  =
  \max\{1,\norm{f}_{0}\norm{g}_{0}
  \norm{f}_{\infty}\norm{g}_{\infty}\},
  \qquad
  H_0=\max\{2,\widehat H\}.
\]
Every coefficient of \(h\) is a sum of products \(f_ag_b\), hence
\[
  \norm{h}_{\infty}\le \widehat H\le H_0 .
\]
Moreover \(\norm{f}_{0},\norm{g}_{0},\norm{h}_{0}\le D+1\) and
\(\norm{f}_{\infty},\norm{g}_{\infty}\le H\), so
\[
  \widehat H\le (D+1)^2H^2 .
\]
This implies
\[
  \log(DH_0)=O(\log(DH)).
\]

We apply Theorem 3.4 of \cite{Giorgi2022-ISSAC} to the unknown polynomial \(\Phi=h\),
using degree bound \(D\), height bound \(H_0\), and sparsity bound \(T\).  In
the parameter range used here, Theorem 3.4 gives the Monte Carlo algorithm
\(\Interpolate\) for an unknown integer polynomial from modular black-box
access, provided valid bounds for its degree, height, and sparsity are supplied.
Its non-probe work is
\[
  \softO\bigl(T(\log D+\log H_0)\bigr).
\]

For any modulus \(m\ge 1\) and residue \(\theta\), define
\[
  \mathcal B_h(\theta,m)
  =
  \left(\sum_{a\in A} f_a\theta^a\right)
  \left(\sum_{b\in B} g_b\theta^b\right)\bmod m .
\]
Evaluation modulo \(m\) is a ring homomorphism \(\Z[x]\to \Z/m\Z\), and
therefore
\[
  \mathcal B_h(\theta,m)\equiv f(\theta)g(\theta)
  \equiv h(\theta)\pmod m .
\]
Thus \(\mathcal B_h\) is a valid modular black box for \(h\).

It remains to account for the cost of answering the structured probes used by
Theorem 3.4.  In an interpolation round with residual sparsity bound \(U\), the
algorithm asks for values of \(h-h^*\), where \(h^*\) is its current partial
reconstruction, on the following root-of-unity lists:
\[
  1,\omega,\omega^2,\ldots,\omega^{2U-1}
  \quad\text{in }\mathbb F_q,
\]
\[
  1,\omega_k,\omega_k^2,\ldots,\omega_k^{U-1}
  \quad\text{in }R=\Z/q^{2k}\Z,
\]
and
\[
  (1+q^k)\omega_k^i,\qquad 0\le i<U,
  \quad\text{again in }R.
\]
The values of \(h^*\) are handled internally by the interpolation algorithm.
We supply the values of \(h\) by evaluating \(f\) and \(g\) on the same lists
and multiplying the corresponding values in the ambient ring.

Let \(P(x)=\sum_{e\in E}c_ex^e\) be either input polynomial, with
\(s=|E|\le T\).  On an unshifted root-of-unity progression,
\[
  P(\omega^i)=\sum_{e\in E} c_e(\omega^e)^i .
\]
This is a transposed-Vandermonde product.  Fact 3.2 of \cite{Giorgi2022-ISSAC} computes
such consecutive sparse root-of-unity values in
\[
  \softO((s+N)\log_2 p)
\]
ring operations for \(N\) requested values after the standard blocking
argument: if \(N>s\), the list is handled in blocks, and repeated nodes are
combined before or during the matrix-vector product.  The same argument applies
over the lifted ring \(R=\Z/q^{2k}\Z\).

For the shifted list, set \(\alpha=1+q^k\).  For every exponent \(e\),
\[
  \alpha^e=(1+q^k)^e\equiv 1+e q^k\pmod {q^{2k}},
\]
because all binomial terms of order at least two are divisible by \(q^{2k}\).
Hence
\[
  P(\alpha\omega_k^i)
  =
  \sum_{e\in E} c_e(1+e q^k)(\omega_k^e)^i
  \quad\text{in }R,
\]
which is again the same transposed-Vandermonde evaluation problem, now with
modified coefficients.  Preparing these coefficients costs only
\(\softO(s\log(DH_0))\) bit operations in the rings used by Theorem 3.4.

Therefore, in one interpolation round, all probe values for both inputs and
their pointwise products are obtained in
\[
  \softO((\norm{f}_{0}+\norm{g}_{0}+U)\log p)
\]
ring operations, with the element bit sizes prescribed in the proof of
Theorem 3.4.  The source theorem chooses parameters satisfying
\[
  p=\softO\bigl(T\log(DH_0)\bigr),
  \qquad
  k\log q=\softO(\log(DH_0)).
\]
The residual bounds \(U\) decrease geometrically, so their sum over all rounds
is \(O(T)\).  The fixed input sparsities \(\norm{f}_{0},\norm{g}_{0}\le T\) are
paid in only a logarithmic number of rounds, which is absorbed by the soft-O
notation.  Thus the total cost of all black-box probe answers, pointwise
products, reductions, and shifted coefficient preparations is
\[
  \softO\bigl(T\log(DH_0)\bigr)
  =
  \softO\bigl(T\log(DH)\bigr).
\]
The non-probe work of Theorem 3.4 and the cost of writing the sparse output are
within the same bound.

All degree, height, and sparsity preconditions for Theorem 3.4 have now been
met:
\[
  D\ge \deg h,\qquad H_0\ge \norm{h}_{\infty},\qquad
  T\ge \norm{h}_{0}.
\]
Consequently, outside the failure event allowed by Theorem 3.4, the
interpolation algorithm returns the sparse representation of \(h=fg\).  This
is exactly the stated Monte Carlo guarantee.  The supplied-\(T\) theorem does
not claim a certified or Las Vegas algorithm.
\end{proof}

\subsection{General Sparse Polynomial Multiplication}
\label{sec-32}

In this section, we show how to achieve the general results based on Theorem \ref{th-ql1}.

\paragraph{If the output sparsity is not supplied.}
The theorem above uses a valid supplied value of \(T\).  If \(\norm{fg}_{0}\)
is not known in advance, one may instead try successive guesses for the output
sparsity and run the same interpolation method with each guess.  Such a wrapper
must verify each candidate product with an appropriate sparse product
verification routine, such as the product verification methods of
\cite{GIORGI202398-JSC}; without that verification step, Theorem 3.4 of
\cite{Giorgi2022-ISSAC} alone does not provide a stopping criterion for the
unknown-output-sparsity setting.  This paragraph is only a conditional wrapper
comment and is not part of the supplied-output-sparsity theorem above.

The following algorithm is used for verification: it computes different candidates for $fg$ with a growing sparsity bound and verifies each candidate using $\text{Sparse product verification}$.

\begin{algorithm}[H]
\caption{Sparse product verification~\citep[Algorithm 9]{GIORGI202398-JSC}}
\begin{algorithmic}[1]
  \Require Sparse polynomials \(f,g,h\in\mathbb{Z}[X]\) with \(\deg(h) \le \deg(f) + \deg(g)\), and a parameter \(0<\epsilon<1\).
  
  \Ensure \texttt{True} if \(h = fg\); otherwise \texttt{False} with probability at least \(1-\epsilon\).
\end{algorithmic}
\end{algorithm}

\begin{corollary}\cite{GIORGI202398-JSC}
Let \(f, g, h \in \mathbb{Z}[X]\) of degree at most \(D\), with norm  \(H = \max\bigl(\|f\|_{\infty}, \|g\|_{\infty}, \|h\|_{\infty}\bigr)\) and sparsity  \(T=\max(\norm{f}_{0},\norm{g}_{0},\norm{h}_{0})\). Then Algorithm 2 has bit complexity
\[
\softO\bigl(T(\log D + \log H)\polylog(1/\epsilon) \bigr).
\]
\end{corollary}

We present a probabilistic algorithm for multiplying two sparse polynomials with integer coefficients when the sparsity \( \|fg\|_0\) of the product is unknown. The algorithm successively guesses \(t = 2^k\) for \(k = 0,1,2,\dots\), invokes a base multiplication routine (which succeeds with probability \(2/3\) if the guess is correct), and verifies each candidate product using a sparse product verification routine. By repeating the base multiplication \(\lceil \log_2(2/\varepsilon)\rceil\) times for each guess and using a final verification step, we drive the overall error probability below any prescribed \(\varepsilon > 0\). The complexity is quasi-linear in the input size and poly-logarithmic in \(1/\varepsilon\).


\begin{algorithm}[H]
\caption{Multiplication of sparse polynomials with integer coefficients (unknown sparsity)}
\begin{algorithmic}[1] 
  \Require Sparse polynomials \(f,g\in\mathbb{Z}[x]\), a target error bound \(0<\varepsilon<1\), and the degree \(D=\deg(fg)\) (or an upper bound on it).
  
  \Ensure The sparse representation of \(fg\) with probability at least \(1-\varepsilon\).
  
  \State Set \(\delta \leftarrow \frac{\varepsilon}{2}\).
  
  \State Compute \(k_{\max} = \bigl\lceil \log_2\bigl( \|f\|_0 \|g\|_0 \bigr) \bigr\rceil\). 
  
  \State Compute \(r = \bigl\lceil \log_2(1/\delta) \bigr\rceil\) and set \(\beta \leftarrow \frac{\varepsilon}{2(k_{\max}+1)r}\).
  
  \For{each \(k = 0, 1, 2, \dots, k_{\max}\)}
    \State Set \(t \leftarrow 2^k\).
    \For{\(r\) times}
      \State Run the base multiplication algorithm (Algorithm 1) with sparsity bound \(t\) and degree bound \(D\) to obtain a candidate product \(\tilde{h}\).
      \If{the verification algorithm (Algorithm 2) called with parameters \((f,g,\tilde{h}, \beta)\) returns \texttt{True}}
        \State Output \(\tilde{h}\) and stop.
      \EndIf
    \EndFor
  \EndFor
  
  \State If no candidate passes verification, return \texttt{Failure}. 
\end{algorithmic}
\end{algorithm}

\subsubsection{Probability Analysis}

Let \(\varepsilon \in (0,1)\) be the desired overall error bound.  
Define \(\delta = \varepsilon/2\) and let \(k_{\max}\) and \(r\) be as in the algorithm, with
\[
\beta = \frac{\varepsilon}{2(k_{\max}+1)r}.
\]

Let \(T = \max(\|f\|_0, \|g\|_0, \|fg\|_0))\).  
There exists a unique \(k^*\) such that \(2^{k^*-1} < T \le 2^{k^*}\) (or \(k^*=0\) if \(t^*=1\)).

\paragraph{Correct guess.}
For the correct guess \(t = 2^{k^*} \ge T\), the base multiplication algorithm (Algorithm 1) succeeds with probability at least \(2/3\).  
Repeating it \(r\) times independently, the probability that all trials fail is at most
\[
(1 - 2/3)^r = (1/3)^r \le 2^{-r} \le \delta.
\]
Hence, with probability at least \(1-\delta\), at least one trial produces the correct product \(fg\).

\paragraph{Verification error.}
The verification algorithm (Algorithm 2) with error bound \(\beta\) returns \texttt{True} for an incorrect candidate with probability at most \(\beta\).  
There are at most \((k_{\max}+1) \cdot r\) verification calls in total. By the union bound, the probability that \emph{any} incorrect candidate passes verification is at most
\[
(k_{\max}+1) \cdot r \cdot \beta = \frac{\varepsilon}{2}.
\]

\paragraph{Total failure probability.}
The algorithm fails if either:
\begin{itemize}
  \item the correct guess never produces a correct candidate (probability \(\le \delta = \varepsilon/2\)), or
  \item an incorrect candidate passes verification and is output (probability \(\le \varepsilon/2\)).
\end{itemize}
By the union bound, the overall failure probability is at most \(\varepsilon\). Therefore,
\[
\Pr[\text{output} = fg] \ge 1 - \varepsilon.
\]

\subsubsection{Complexity Analysis}

\paragraph{Base multiplication.}
For a guess $t = 2^k$, if $2^k \le T$, we may use the true sparsity $T$ as an upper bound for the complexity; if $2^k > T$, the complexity is bounded by $\softO(2^k (\log D + \log H))$.

\paragraph{Verification.}
Similarly, when $2^k \le T$, the verification complexity can be bounded using $T$ as $\softO(T (\log D + \log H)\polylog(1/\beta))$; when $2^k > T$, it is bounded by $\softO(2^k (\log D + \log H)\polylog(1/\beta))$.
Given
\(
1/\beta = O\!\left( \frac{\log T \cdot \log(1/\varepsilon)}{\varepsilon} \right),
\)
substituting this yields the complexity in terms of the original parameters:
\[
\softO\!\left( T (\log D + \log H) \cdot \operatorname{polylog}\!\left(1/\varepsilon \right) \right),
\]
and, when \( 2^k > T \),
\[
\softO\!\left( 2^k (\log D + \log H) \cdot \operatorname{polylog}\!\left( 1/\varepsilon \right) \right).
\]

\paragraph{Average-case complexity analysis.}
Consider the moment when the guess $2^k$ first exceeds the true sparsity $T$, i.e., $k = k^* + 1$.  
If the algorithm stops at this step, the total computational cost is 
\[
C_{k^*+1} = \softO\!\bigl(T (\log D + \log H) \cdot \operatorname{polylog}\!\left(1/\varepsilon \right)+ 2^{k^*+1} \cdot r \cdot (\log D + \log H) \bigr).
\]
If it instead stops at $k = k^* + 2$, then all previous guesses (including $k = k^*+1$) must have failed to produce a correct candidate that passes verification.  
The probability that a given guess fails to stop is at most $\varepsilon/2$ , so the probability of reaching step $k^*+2$ is at most $\varepsilon/2$.  
More systematically, let $p$ be the failure probability per guess step. Then the probability that the algorithm first stops at $k = k^* + s$ is at most $(\varepsilon/2)^{s-1}$.

Thus the expected cost is bounded by
\[
\mathbb{E}[\text{cost}] \le T (\log D + \log H) \cdot \operatorname{polylog}\!\left(1/\varepsilon \right)+ \sum_{s=1}^{\infty} \bigl( (\varepsilon/2)^{s-1} \bigr) \cdot \softO\!\bigl( 2^{k^*+s} \cdot r \cdot (\log D + \log H)\operatorname{polylog}\!\left(1/\varepsilon \right) \bigr).
\]
Since $2^{k^*+s} = T \cdot 2^s$, the geometric series converges:
\[
\mathbb{E}[\text{cost}] \le \softO\!\left( T \cdot  (\log D + \log H) \cdot \operatorname{polylog}\!\left(1/\varepsilon \right)\cdot \sum_{s=1}^{\infty} \varepsilon^{s-1}\right).
\]
As  $\sum_{t=1}^{\infty} (\varepsilon)^{t-1} = \frac{1}{1-\varepsilon} = O(1)$.  
Therefore,
\[
\mathbb{E}[\text{cost}] = \softO\!\left( T \cdot  (\log D + \log H) \cdot \operatorname{polylog}\!\left(1/\varepsilon \right)\right).
\]

This shows that the average-case bit complexity matches the worst-case bound up to constants, and the geometric decay ensures that the contribution from large $k$ is negligible.

Hence, the algorithm achieves quasi-linear bit complexity in the input size without prior knowledge of \(T\), with an arbitrarily high success probability \(1-\varepsilon\).

\paragraph{Multivariate case.}

Using the Kronecker substitution \(x_i \mapsto X^{D^{i-1}}\) (with \(D\) larger than any exponent), a multivariate multiplication \(fg\) reduces to a univariate multiplication of two polynomials of degree \(\le D^n\).  
Applying the univariate algorithm and substituting back yields the same guaranties.  
The bit complexity becomes
\[
\tilde{\mathcal{O}}\!\left( T \cdot (n \log D + \log H)  \operatorname{polylog}\!\left(1/\varepsilon \right)\right),
\]
where \(T = \max(\#f,\#g,\#fg)\), \(D\) bounds the total degree, and \(C\) bounds the coefficients.  
This is optimal for integer coefficients up to polylogarithmic factors.

\section{Multiplication with Finite Field Coefficients}
\label{sec-Mult-F}
In this section, we prove Theorem \ref{th-m2}. Similar to Section \ref{sec-1}, we first prove the supplied output sparsity case in Theorem \ref{thm:fq-product} and the extend the result to the general case in Remark \ref{rem-ff1}.
%
%

The separated-variable \(\SoftO\) convention used below hides
polylogarithmic factors in \(T\), in iterated logarithms of the displayed degree
parameter, and in \(1/\delta\), but it does not hide an additional visible power
of \(L_D\).

\begin{theorem}[Large-characteristic finite-field sparse product: Supplied output sparsity]
\label{thm:fq-product}
Let \(f,g\in\Fq[x]\), put \(h=fg\), and suppose that
\[D = \max(\deg(f),\deg(g)),
  T\ge \max(\|f\|_0, \|g\|_0, \|fg\|_0)
\]
is supplied.  Let \(D_{\rm out}\) be a known upper bound for \(\deg h\), let
\(0<\delta<1\), and assume that
\[
  \charac(\Fq)>D_{\rm out}.
\]

Then there is a Monte Carlo algorithm that returns the exact sparse product
\(h\) with a probability of at least \(1-\delta\), using
\[
  \SoftO\!\left(T\log D\,\polylog(1/\delta)\right)
\]
unit-cost field operations in \(\Fq\).  In an explicit finite-field
representation, the corresponding bit complexity is
\[
  \SoftO\!\left(T\log D\log q\,\polylog(1/\delta)\right).
\]
\end{theorem}


\begin{remark}
\label{rem-ff1}
  In the finite field setting $\mathbb{F}_q[x]$, the same framework in Section \ref{sec-32} applies to generate Theorem \ref{thm:fq-product} to the case where the output sparsity is not provided. 
There exists a sparse multiplication algorithm that assumes a known sparsity bound $t$ for the product, as well as a sparse product verification algorithm whose complexity is $\tilde{O}\bigl(T(\log D + \log q)\operatorname{polylog}(1/\delta)\bigr)$ (see Algorithm 9, Corollary 5.15 and Corollary 5.16 in \cite{GIORGI202398-JSC}). Consequently, by following the identical guessing strategy $t = 2^k$ and using the verification routine as a filter, we can eliminate the requirement that the output sparsity $\|fg\|_0$ be provided as input. The resulting probabilistic algorithm achieves expected complexity $\tilde{O}\bigl(T\log D\log q\operatorname{polylog}(1/\delta)\bigr)$ and succeeds with probability at least $1-\delta$, exactly as in the integer coefficient case. 
\end{remark}

\subsection{Interpolation Ingredients}
In this section, we present the components of sparse polynomial interpolation that are employed in the algorithm.
The algorithm uses derivative images in the sparse-interpolation style of
\cite{Huang2019-ISSAC,Giorgi2021-ISSAC}.  The proof needed here is
self-contained: derivative images generate a short candidate list containing
the true support, and ordinary cyclic images then filter that list.

Put \(B=\max\{D_{\rm out}+1,2\}\) and
$
  L_D=\log\max(D_{\rm out}+1,2)
$.  If
\(S\subseteq\{0,\ldots,D_{\rm out}\}\) is finite, a prime \(\ell\) isolates
\(e\in S\) in \(S\) if no other element of \(S\) is congruent to \(e\) modulo
\(\ell\).

\begin{lemma}[Prime isolation sampler]
\label{lem:prime-isolation}
Let \(S\subseteq\{0,\ldots,D_{\rm out}\}\), let \(|S|\le M\), and let
\(0<\eta<1\).  There is a Monte Carlo sampler using
\[
  k=O(\log(\max(M,2)/\eta))
\]
primes, each of size \(\SoftO(M L_D)\), such that with probability at least
\(1-\eta\), every \(e\in S\) is isolated in \(S\) by at least one sampled
prime.
\end{lemma}

\begin{proof}
Choose
\[
  N=\lceil 8\max(M,2)\lceil\log_2 B\rceil\rceil
\]
and sample
\[
  k=C\lceil\log(\max(M,2)/\eta)\rceil
\]
primes independently and uniformly from the first \(N\) primes, where \(C\) is
an absolute constant large enough for the probability estimate below.  Fix
\(e\in S\).  If a sampled prime \(\ell\) fails to isolate \(e\), then
\(\ell\mid(e-e')\) for some \(e'\in S\setminus\{e\}\).  For fixed \(e'\), the
nonzero difference \(e-e'\) has absolute value less than \(B\), so it has at
most \(\log_2 B\) distinct prime divisors.  Therefore at most
\((M-1)\log_2 B\) among the first \(N\) primes are bad for this fixed \(e\).
By the definition of \(N\), a uniformly sampled prime from this list is bad
for \(e\) with probability at most \(1/4\).  Hence the probability that all
\(k\) sampled primes are bad for \(e\) is at most \(4^{-k}\), and the constant
\(C\) makes this at most \(\eta/\max(M,2)\).  A union bound over all
\(e\in S\) gives the claimed probability.  The \(N\)-th prime has size
\(\SoftO(M L_D)\) by the standard prime-number estimate for the growth of the
\(N\)-th prime \cite{HardyWright2008}.
\end{proof}

\begin{lemma}[Derivative candidate generation]
\label{lem:derivative-candidates}
Let
\[
  P=\sum_{e\in S} c_e x^e,\qquad
  S\subseteq\{0,\ldots,D_{\rm out}\},
\]
over \(K=\Fq\), with \(\charac(K)>D_{\rm out}\).  Suppose
\[
  P_\ell=P\bmod(x^\ell-1),\qquad
  P'_\ell=P'\bmod(x^\ell-1)
\]
are known.  One can form a candidate set
\(C_\ell\subseteq\{0,\ldots,D_{\rm out}\}\) with \(|C_\ell|\le |S|\), and
every term isolated by \(\ell\) has its exponent in \(C_\ell\).
\end{lemma}

\begin{proof}
For each residue \(r\) whose coefficient \(a_r\) in \(P_\ell\) is nonzero, let
\(b_r\) be the coefficient of \(x^{r-1\bmod \ell}\) in \(P'_\ell\).  Compute
\(\alpha_r=b_r/a_r\) in \(K\), and insert \(d\) into \(C_\ell\) when
\(\alpha_r\) is the prime-subfield image of a unique
\[
  d\in\{0,\ldots,D_{\rm out}\}
\]
with \(d\equiv r\pmod\ell\).

There is at most one inserted candidate for each nonzero bucket.  Each nonzero
bucket contains at least one element of \(S\), so the number of inserted
candidates is at most \(|S|\).  If \(e\in S\) is isolated modulo \(\ell\), then
the bucket \(r=e\bmod\ell\) has \(a_r=c_e\ne0\).  The derivative bucket at
\(r-1\bmod\ell\) has \(b_r=e c_e\), including the case \(e=0\), where
\(b_r=0\).  Since \(\charac(K)>D_{\rm out}\), the images of
\(0,\ldots,D_{\rm out}\) in the prime subfield are distinct.  Thus \(b_r/a_r\)
decodes uniquely to \(e\), and every isolated true exponent is inserted into
\(C_\ell\).
\end{proof}

\begin{lemma}[Candidate filtering]
\label{lem:candidate-filter}
Let \(P=\sum_{e\in S}c_ex^e\in\Fq[x]\), with
\(S\subseteq C\subseteq\{0,\ldots,D_{\rm out}\}\), where \(C\) is known.  With
probability at least \(1-\eta\),
\(\SoftO(\log(\max(|C|,2)/\eta))\) ordinary cyclic images \(P_\ell\), with
primes of size \(\SoftO(\max(|C|,1)L_D)\), recover exactly the terms of \(P\).
\end{lemma}

\begin{proof}
If \(C=\varnothing\), then \(S=\varnothing\), so returning the zero polynomial
is correct. Otherwise, apply Lemma~\ref{lem:prime-isolation} to the known set
\(C\) with failure parameter \(\eta\).  On the isolation event, every \(e\in C\)
has a sampled prime \(\ell\) that isolates \(e\) inside \(C\).  Choose such an
\(\ell\) using the known exponents in \(C\).  Because \(S\subseteq C\), the
coefficient in the \(e\bmod\ell\) bucket of \(P_\ell\) is \(c_e\) when
\(e\in S\), and it is \(0\) when \(e\notin S\).  Keeping exactly those
candidates whose selected bucket coefficient is nonzero returns exactly \(P\).
The failure probability is at most \(\eta\).
\end{proof}

\subsection{Proof of Theorem \ref{thm:fq-product}}

For a prime \(\ell\), define
\[
  f_\ell=f\bmod (x^\ell-1),\qquad
  g_\ell=g\bmod (x^\ell-1).
\]
Because reduction modulo \(x^\ell-1\) is a ring homomorphism, the product image
is
\[
  h_\ell=f_\ell g_\ell\bmod (x^\ell-1).
\]
Also define
\[
  f'_\ell=f'\bmod (x^\ell-1),\qquad
  g'_\ell=g'\bmod (x^\ell-1).
\]
The formal derivative satisfies \((fg)'=f'g+fg'\), and reducing this identity
modulo \(x^\ell-1\) gives
\[
  h'_\ell=f'_\ell g_\ell+f_\ell g'_\ell \pmod {x^\ell-1}.
\]
Thus \(h_\ell\) and \(h'_\ell\) are obtained from the sparse inputs by folding
\(f,g,f',g'\) to length \(\ell\), performing dense cyclic products, and adding
the two derivative-product terms.  No inverse of \(g_\ell\) is used.

\begin{proof}[Proof of Theorem~\ref{thm:fq-product}]
If \(f=0\) or \(g=0\), return zero.  Otherwise let \(S=\supp h\).  Apply
Lemma~\ref{lem:prime-isolation} to the unknown set \(S\), with \(M=T\) and
failure budget \(\delta/2\).  The sampler itself uses only \(T\),
\(D_{\rm out}\), and \(\delta\).  For each sampled prime \(\ell\), compute
\(h_\ell\) and \(h'_\ell\) by the cyclic product and derivative formulas above,
then use Lemma~\ref{lem:derivative-candidates} to form \(C_\ell\).  Let
\[
  C=\bigcup_\ell C_\ell .
\]
On the isolation event from Lemma~\ref{lem:prime-isolation}, every element of
\(S\) lies in \(C\).  Also each \(C_\ell\) has size at most \(|S|\le T\), and
the number of sampled primes is \(O(\log(T/\delta))\), so
\[
  |C|=\SoftO(T\,\polylog(1/\delta))
\]
with no additional visible power of \(L_D\).

Run Lemma~\ref{lem:candidate-filter} on the known set \(C\) with failure budget
\(\delta/2\), using ordinary product images \(h_\ell=f_\ell g_\ell\).  If both
the generation-stage and filtering-stage isolation events occur, the output is
exactly \(h\).  The union bound gives total failure probability at most
\(\delta\).

It remains to record the cost.  For one modulus \(\ell\), folding the sparse
inputs and, during candidate generation, their formal derivatives costs
\(\SoftO(T)\) field operations plus exponent bookkeeping with \(L_D\)-bit
exponents.  Dense cyclic multiplication of length \(\ell\) costs
\(\SoftO(\ell)\) operations over the coefficient algebra by the
Cantor--Kaltofen dense multiplication theorem \cite{CantorKaltofen1991}.
Lemma~\ref{lem:prime-isolation} samples primes of size \(\SoftO(TL_D)\) in the
generation stage and \(\SoftO(|C|L_D)\) in the filtering stage.  The number of
sampled primes is polylogarithmic in \(T\) and \(1/\delta\).  Summing the dense
lengths and the sparse folding and bookkeeping costs gives
\[
  \SoftO\!\left(TL_D\,\polylog(1/\delta)\right)
\]
field operations under the separated-variable convention.  In an explicit
representation of \(\Fq\), field operations and prime-subfield decoding are
charged with a factor \(\log q\) up to hidden polylogarithmic factors, giving
the bit complexity
$
  \SoftO\!\left(TL_D\log q\,\polylog(1/\delta)\right).$
\end{proof}

\section{Conclusion}
\label{sec-Conc}

We first provide a counterexample to a key result of \cite{Giorgi2020-ISSAC},  thereby invalidating the quasi-linear sparse polynomial multiplication claim.
By applying the MBB interpolation algorithm, we can derive a quasi-linear algorithm for multiplying sparse polynomials with integer coefficients.
Finally, we introduce a sparse polynomial multiplication algorithm defined over the large characteristic finite field, which attains a bit complexity of $\widetilde{O}(T \log D \log H)$.

\section*{Acknowledgment}
This paper is supported by the Strategic Priority Research Program of CAS Grants XDA0480502 and XDA0480503, NSFC Grants 12288201 and 92270001.

\bibliographystyle{KLMM/klmm}   
\bibliography{Refs}

\begin{thebibliography}{22}
\providecommand{\natexlab}[1]{#1}
\providecommand{\url}[1]{\texttt{#1}}
\expandafter\ifx\csname urlstyle\endcsname\relax
  \providecommand{\doi}[1]{doi: #1}\else
  \providecommand{\doi}{doi: \begingroup \urlstyle{rm}\Url}\fi

\bibitem[Agrawal et~al.(2004)Agrawal, Kayal, and Saxena]{agrawal2004primes}
Manindra Agrawal, Neeraj Kayal, and Nitin Saxena.
\newblock Primes is in p.
\newblock \emph{Annals of Mathematics}, 160\penalty0 (2):\penalty0 781--793, 2004.
\newblock ISSN 0003486X.
\newblock URL \url{http://www.jstor.org/stable/3597229}.

\bibitem[Arnold(2016)]{Thesis-Arnold}
Andrew Arnold.
\newblock \emph{Sparse Polynomial Interpolation and Testing (PhD Thesis)}.
\newblock University of Waterloo, 2016.

\bibitem[Cantor and Kaltofen(1991)]{CantorKaltofen1991}
David~G. Cantor and Erich Kaltofen.
\newblock On fast multiplication of polynomials over arbitrary algebras.
\newblock \emph{Acta Inf.}, 28\penalty0 (7):\penalty0 693–701, October 1991.
\newblock ISSN 0001-5903.
\newblock \doi{10.1007/BF01178683}.
\newblock URL \url{https://doi.org/10.1007/BF01178683}.

\bibitem[Cole and Hariharan(2002)]{Cole2002STOC}
Richard Cole and Ramesh Hariharan.
\newblock Verifying candidate matches in sparse and wildcard matching.
\newblock In \emph{Proceedings of the Thiry-Fourth Annual ACM Symposium on Theory of Computing}, STOC '02, pp.\  592–601, New York, NY, USA, 2002. Association for Computing Machinery.
\newblock ISBN 1581134959.
\newblock \doi{10.1145/509907.509992}.
\newblock URL \url{https://doi.org/10.1145/509907.509992}.

\bibitem[Cook and Aanderaa(1969)]{Cook1966}
Stephen~A. Cook and Stål~O. Aanderaa.
\newblock On the minimum computation time of functions.
\newblock \emph{Transactions of the American Mathematical Society}, 142:\penalty0 291--314, 1969.
\newblock ISSN 00029947.
\newblock URL \url{http://www.jstor.org/stable/1995359}.

\bibitem[Davenport and Carette(2009)]{SparsityChallenges}
James~Harold Davenport and Jacques Carette.
\newblock The sparsity challenges.
\newblock In \emph{2009 11th International Symposium on Symbolic and Numeric Algorithms for Scientific Computing}, pp.\  3--7, 2009.
\newblock \doi{10.1109/SYNASC.2009.62}.

\bibitem[Fischer(2025)]{NickFischer2025SODA}
Nick Fischer.
\newblock Sumsets, 3sum, subset sum: Now for real!
\newblock In \emph{Proceedings of the 2025 Annual ACM-SIAM Symposium on Discrete Algorithms (SODA)}, pp.\  4520--4546, 2025.
\newblock \doi{10.1137/1.9781611978322.155}.
\newblock URL \url{https://epubs.siam.org/doi/abs/10.1137/1.9781611978322.155}.

\bibitem[Giorgi et~al.(2020)Giorgi, Grenet, and Cray]{Giorgi2020-ISSAC}
Pascal Giorgi, Bruno Grenet, and Armelle Perret~du Cray.
\newblock Essentially optimal sparse polynomial multiplication.
\newblock In \emph{Proceedings of the 45th International Symposium on Symbolic and Algebraic Computation}, ISSAC '20, pp.\  202–209, New York, NY, USA, 2020. Association for Computing Machinery.
\newblock ISBN 9781450371001.
\newblock \doi{10.1145/3373207.3404026}.
\newblock URL \url{https://doi.org/10.1145/3373207.3404026}.

\bibitem[Giorgi et~al.(2021)Giorgi, Grenet, and Perret~du Cray]{Giorgi2021-ISSAC}
Pascal Giorgi, Bruno Grenet, and Armelle Perret~du Cray.
\newblock On exact division and divisibility testing for sparse polynomials.
\newblock In \emph{Proceedings of the 2021 International Symposium on Symbolic and Algebraic Computation}, ISSAC '21, pp.\  163–170, New York, NY, USA, 2021. Association for Computing Machinery.
\newblock ISBN 9781450383820.
\newblock \doi{10.1145/3452143.3465539}.
\newblock URL \url{https://doi.org/10.1145/3452143.3465539}.

\bibitem[Giorgi et~al.(2022)Giorgi, Grenet, Perret~du Cray, and Roche]{Giorgi2022-ISSAC}
Pascal Giorgi, Bruno Grenet, Armelle Perret~du Cray, and Daniel~S. Roche.
\newblock Sparse polynomial interpolation and division in soft-linear time.
\newblock In \emph{Proceedings of the 2022 International Symposium on Symbolic and Algebraic Computation}, ISSAC '22, pp.\  459–468, New York, NY, USA, 2022. Association for Computing Machinery.
\newblock ISBN 9781450386883.
\newblock \doi{10.1145/3476446.3536173}.
\newblock URL \url{https://doi.org/10.1145/3476446.3536173}.

\bibitem[Giorgi et~al.(2023)Giorgi, Grenet, and {Perret du Cray}]{GIORGI202398-JSC}
Pascal Giorgi, Bruno Grenet, and Armelle {Perret du Cray}.
\newblock Polynomial modular product verification and its implications.
\newblock \emph{Journal of Symbolic Computation}, 116:\penalty0 98--129, 2023.
\newblock ISSN 0747-7171.
\newblock \doi{https://doi.org/10.1016/j.jsc.2022.08.011}.
\newblock URL \url{https://www.sciencedirect.com/science/article/pii/S0747717122000773}.

\bibitem[Grenet(2026)]{Grenet-H2026}
Bruno Grenet.
\newblock Fast polynomial computations with space constraints, 2026.
\newblock URL \url{https://arxiv.org/abs/2511.11267}.

\bibitem[Hardy and Wright(2008)]{HardyWright2008}
G.~H. Hardy and E.~M. Wright.
\newblock \emph{An Introduction to the Theory of Numbers}.
\newblock Oxford University Press, 6th edition, 2008.

\bibitem[Huang(2019)]{Huang2019-ISSAC}
Qiao-Long Huang.
\newblock Sparse polynomial interpolation over fields with large or zero characteristic.
\newblock In \emph{Proceedings of the 2019 International Symposium on Symbolic and Algebraic Computation}, ISSAC '19, pp.\  219–226, New York, NY, USA, 2019. Association for Computing Machinery.
\newblock ISBN 9781450360845.
\newblock \doi{10.1145/3326229.3326250}.
\newblock URL \url{https://doi.org/10.1145/3326229.3326250}.

\bibitem[Jin and Xu(2024)]{Jin2025STOC}
Ce~Jin and Yinzhan Xu.
\newblock Shaving logs via large sieve inequality: Faster algorithms for sparse convolution and more.
\newblock In \emph{Proceedings of the 56th Annual ACM Symposium on Theory of Computing}, STOC 2024, pp.\  1573–1584, New York, NY, USA, 2024. Association for Computing Machinery.
\newblock ISBN 9798400703836.
\newblock \doi{10.1145/3618260.3649605}.
\newblock URL \url{https://doi.org/10.1145/3618260.3649605}.

\bibitem[Karatsuba and Ofman(1962)]{KaratsubaOfman1962}
Anatolii~Alekseevich Karatsuba and Yu~P Ofman.
\newblock Multiplication of many-digital numbers by automatic computers.
\newblock In \emph{Doklady Akademii Nauk}, volume 145, pp.\  293--294. Russian Academy of Sciences, 1962.

\bibitem[{MechMath Team}(2026)]{MMAT}
{MechMath Team}.
\newblock Mechmath agent team.
\newblock Academy of Mathematics and Systems Science, Chinese Academy of Sciences, 2026.
\newblock URL \url{https://eonmath.github.io/mechmath}.

\bibitem[Nakos(2020)]{9076060}
Vasileios Nakos.
\newblock Nearly optimal sparse polynomial multiplication.
\newblock \emph{IEEE Transactions on Information Theory}, 66\penalty0 (11):\penalty0 7231--7236, 2020.

\bibitem[Roche(2018)]{Roche2018}
Daniel~S. Roche.
\newblock What can (and can't) we do with sparse polynomials?
\newblock In \emph{Proceedings of the 2018 ACM International Symposium on Symbolic and Algebraic Computation}, ISSAC '18, pp.\  25–30, New York, NY, USA, 2018. Association for Computing Machinery.
\newblock ISBN 9781450355506.
\newblock \doi{10.1145/3208976.3209027}.
\newblock URL \url{https://doi.org/10.1145/3208976.3209027}.

\bibitem[Sch{\"o}nhage and Strassen(1971)]{SchonhageStrassen1971}
A.~Sch{\"o}nhage and V.~Strassen.
\newblock Schnelle multiplikation gro{\ss}er zahlen.
\newblock \emph{Computing}, 7\penalty0 (3):\penalty0 281--292, 1971.
\newblock ISSN 1436-5057.
\newblock \doi{10.1007/BF02242355}.
\newblock URL \url{https://doi.org/10.1007/BF02242355}.

\bibitem[Toom(1963)]{Toom1963}
A.~L. Toom.
\newblock The complexity of a scheme of functional elements realizing the multiplication of integers.
\newblock \emph{Soviet Mathematics Doklady}, 3:\penalty0 714--716, 1963.

\bibitem[von~zur Gathen and Gerhard(2013)]{ModernCA}
Joachim von~zur Gathen and J{\"u}rgen Gerhard.
\newblock \emph{Modern Computer Algebra}.
\newblock Cambridge University Press, 2013.

\end{thebibliography}

\appendix
\clearpage

\section{Lean Formalization Details}
\label{app:lean-formalization}

This appendix summarizes Lean development at the level of mathematical
structure rather than source code.  Ordinary algebraic, combinatorial, and
cost-envelope obligations are proved inside the project.  Large external
algorithmic theorems and a small number of probability-library obligations are
exposed as named axioms.

\subsection{Files and Formal Targets}
In this subsection, we briefly introduce the content in each Lean file under \texttt{SparseProduct}.

\texttt{Foundations.lean} contains the shared definitions: polynomial sparsity \texttt{Polynomial.sparsity}, $T$-sparsity \texttt{Polynomial.IsTSparse}, integer coefficient height \texttt{Polynomial.linftyNorm}, the computable height bound \texttt{Polynomial.heightBound}, and the soft-O predicates \texttt{IsSoftO} and \texttt{SoftO}.

\texttt{SuppliedOutputSparsity.lean} contains the Lean counterpart of Theorem~\ref{th-ql1}.  
Specifically, its main theorem is \texttt{suppliedOutputSparsityTheorem}.

\texttt{SparseProductCyclicImage.lean} contains the finite-field construction used for Theorem~\ref{thm:fq-product}: cyclic images, derivative images, prime isolation, candidate generation, candidate filtering, and the two top-level result statements \texttt{SparseProduct.SparseProductCyclicImage.finiteFieldSparseProduct} and \texttt{SparseProduct.SparseProductCyclicImage.integerSparseProductReduction}.

\texttt{RandomizedAlgorithms.lean} contains named black-box Monte Carlo statements for finite-field sparse multiplication, randomized prime generation, and the integer reduction by finite-field lifting.

\subsection{Supplied Output-Sparsity Theorem}

The Lean formalization of Theorem~\ref{th-ql1} starts from the definitions
\[
  \texttt{Polynomial.sparsity},\qquad
  \texttt{Polynomial.linftyNorm},\qquad
  \texttt{Polynomial.heightBound}.
\]
It then proves the basic algebraic reductions used in the paper.  The theorem
\texttt{monomialProduct} covers the one-term case.  The theorem
\texttt{heightBound\_spec} proves that every coefficient of \(fg\) is bounded by
\[
  \norm{f}_0\norm{g}_0\norm{f}_\infty\norm{g}_\infty,
\]
and hence by the supplied height bound \(H_0\).  The theorem
\texttt{deg\_product} records the degree decomposition over \(\mathbb Z[x]\),
and \texttt{log\_DH0\_eq\_log\_DH} formalizes the replacement of
\(\log(DH_0)\) by \(\log(DH)\) inside soft-\(O\) notation.

The modular black-box used by the interpolation theorem is represented by
\texttt{modularBlackBox}; its correctness is proved in
\texttt{modularBlackBox\_correct}, using evaluation modulo \(m\) as a ring
homomorphism.  The shifted probes used in the modular-black-box interpolation
proof are covered by \texttt{shifted\_root\_identity} and
\texttt{shifted\_evaluation}.  Finally, the probe accounting is formalized
through \texttt{probeCost\_one\_round}, the geometric summation lemma
\texttt{sum\_div\_pow\_two\_le\_two\_mul}, and
\texttt{probeCost\_total}.

The top-level theorem
\[
  \texttt{SuppliedOutputSparsity.suppliedOutputSparsityTheorem}
\]
proves the soft-\(O\) bit-complexity conclusion corresponding to
Theorem~\ref{th-ql1}.  The interpolation theorem itself is not reproved:
Theorem~3.4 of Giorgi--Grenet--Perret du Cray--Roche~\cite{Giorgi2022-ISSAC}
is imported as \texttt{ggprTheorem34}, and Fact~3.2 of the same paper is
imported as \texttt{ggprFact32}.

\subsection{Finite-Field Theorem}

For Theorem~\ref{thm:fq-product}, the formalization separates the algebraic
image identities from the randomized algorithmic wrapper.  The cyclic image
\(P\bmod (x^\ell-1)\) is defined as
\[
  \texttt{SparseProduct.SparseProductCyclicImage.cyclicImage}.
\]
The theorems \texttt{cyclicImage\_mul} and
\texttt{cyclicImage\_derivative\_mul} formalize the identities
\[
  (fg)_\ell=f_\ell g_\ell \bmod (x^\ell-1),
  \qquad
  (fg)'_\ell=f'_\ell g_\ell+f_\ell g'_\ell \bmod (x^\ell-1).
\]
The coefficient formula for cyclic images is proved as
\texttt{cyclicImage\_coeff}.

The prime-isolation part is represented by the predicate
\texttt{IsolatesExponent}.  The difference bound
\texttt{exponent\_diff\_bound}, the prime-divisor bound
\texttt{primeFactors\_card\_le\_log2}, and the bad-prime count
\texttt{badPrimes\_card\_le} are proved.  The existence of a sufficiently
large sampling set of primes and the final probability estimate for uniform
sampling are left as the named axioms \texttt{exists\_sampling\_primeSet} and
\texttt{uniformOfFinset\_isolation\_failure\_bound}.

The generation of derivative-based candidate is implemented by
\texttt{candidateSet}; its specification theorem \texttt{candidateSet\_spec}
proves that the candidate set lies in \(\{0,\ldots,D_{\rm out}\}\), has the
expected cardinality bound, and contains every exponent isolated by the sampled
modulus.  Candidate filtering is summarized by
\texttt{exists\_candidate\_filter\_recovers}; its degenerate singleton
probability case is currently the axiom
\texttt{uniformOfFinset\_singleton\_recover\_bound}.

The top-level finite-field theorem is
\[
  \texttt{SparseProduct.SparseProductCyclicImage.finiteFieldSparseProduct}.
\]
It reduces to the black-box axiom
\[
  \texttt{monteCarlo\_sparseProduct\_finiteField}.
\]
The unit-cost field-operation variant is
\[
  \texttt{monteCarlo\_sparseProduct\_finiteField\_fieldOps},
\]
and the field-operation to bit-complexity conversion is proved by
\texttt{IsSoftO.mul\_logCard}.  The Lean theorem
\[
  \texttt{SparseProduct.SparseProductCyclicImage.integerSparseProductReduction}
\]
records the corresponding univariate integer reduction through random prime
selection, reduction to a large-characteristic finite field, and balanced
lifting back to \(\mathbb Z[x]\).  This reduction is packaged as the axiom
\texttt{monteCarlo\_sparseProduct\_integer}, with randomized prime generation
supplied by \texttt{monteCarlo\_primeInInterval}.

\subsection{Axiom Boundary}

The current development deliberately uses named axioms rather than hidden proof
holes.  The nine remaining axioms are:
\[
\begin{gathered}
\texttt{ggprTheorem34},\qquad
\texttt{ggprFact32},\qquad
\texttt{monteCarlo\_sparseProduct\_finiteField\_fieldOps},\\
\texttt{monteCarlo\_sparseProduct\_finiteField},\qquad
\texttt{monteCarlo\_primeInInterval},\\
\texttt{monteCarlo\_sparseProduct\_integer},\qquad
\texttt{exists\_sampling\_primeSet},\\
\texttt{uniformOfFinset\_isolation\_failure\_bound},\qquad
\texttt{uniformOfFinset\_singleton\_recover\_bound}.
\end{gathered}
\]
The first two are imported from the modular-black-box interpolation paper of
Giorgi--Grenet--Perret du Cray--Roche.  The next four package randomized
algorithmic statements whose full formalization would require substantial
probability and bit-complexity infrastructure.  The final three package
prime-density and uniform-sampling probability steps used in the prime-isolation
and filtering arguments.  A separate file, \texttt{docs/references.md}, records
the bibliographic sources associated with these axioms.

\end{document}